\documentclass{article}
\usepackage{graphicx,amssymb}
\pdfoutput=1

\usepackage[utf8]{inputenc}
\usepackage{fullpage}
\usepackage{amsthm,mathscinet}
\usepackage[usenames, dvipsnames]{color}
\usepackage[none]{hyphenat}
\usepackage{hyperref}
\usepackage{todonotes}
\usepackage{url}

\usepackage{thm-restate}
    \newtheorem{theorem}{Theorem}
    
    \newtheorem{lemma}[theorem]{Lemma}

    \newtheorem{claim}[theorem]{Claim}

\newcommand{\remove}[1]{{}}

\newcommand{\changed}[1]{{#1}}

\title{Reconstructing a Polyhedron between Polygons in Parallel Slices}
\author{
        Therese Biedl\thanks{Cheriton School of Computer Science, University of Waterloo. Research of TB, VI and AL supported by NSERC. This research was initiated at the Algorithmic Problem Session group at the University of Waterloo} \\
        Anna Lubiw\footnotemark[1]
    \and
        Pavle Bulatovic\footnotemark[1]\\
        Owen Merkel\footnotemark[1]
    \and
        Veronika Irvine\footnotemark[1]\\ 
        Anurag Murty Naredla\footnotemark[1]
}

\begin{document}

\maketitle

\begin{abstract}
Given two $n$-vertex polygons, $P=(p_1, \ldots, p_n)$ lying in the $xy$-plane at $z=0$, and $P'=(p'_1, \ldots, p'_n)$ lying in the $xy$-plane at $z=1$,
a \emph{banded surface} is a triangulated surface homeomorphic to an 
%open cylinder 
annulus connecting $P$ and $P'$ such that the \changed{triangulation's edge set} %surface 
contains vertex disjoint paths $\pi_i$ connecting $p_i$ to $p'_i$ for all $i =1, \ldots, n$.  
The surface then consists of \emph{bands}, where the $i$th band goes between $\pi_i$ and $\pi_{i+1}$. 
We give a polynomial-time algorithm to find a banded surface without Steiner points if one exists.
We explore connections between banded surfaces and linear morphs, where time in the morph corresponds to the $z$ direction.  In particular, we show that if $P$ and $P'$ are convex and the linear morph from $P$ to $P'$ (which moves the $i$th vertex on a straight line from $p_i$ to $p'_i$) %between them
remains planar at all times,
%preserves \changed{planarity}, %convexity, 
then there is a banded surface without Steiner points.

\end{abstract}

%%%%%%%%%%%%%%%%%%%%%%%%%%%%%%%%%%%%%
% Introduction
%%%%%%%%%%%%%%%%%%%%%%%%%%%%%%%%%%%%%
\section{Introduction}

The problem of reconstructing a 3D polyhedral structure between two planar cross-sections has
been heavily studied because of its many practical applications, e.g., in medicine, \changed{for} %, 
constructing models of body organs from MRI slices. 
Most approaches, e.g.~\cite{barequet2004contour}, separate the problem into two steps, both of which are hard and are tackled via heuristics:  
(1) choose a correspondence between the two cross-sections; (2) then 
construct a triangulated surface using extra Steiner points.
The problem is considered to be well-solved by these heuristic methods, but many theoretical questions remain open.
We focus on the second step, i.e., we assume that the correspondence between the two cross-sections is given. 
Also, we focus on the case of two polygons, though the  case  of general planar subdivisions (i.e., planar graph drawings) is also of interest.

There is a close connection between the polyhedron reconstruction  problem and the problem of ``morphing'' or continuously transforming one planar structure to another.  
This connection is explained in more detail later in the Introduction, and motivates our formulation of the polyhedron reconstruction problem.

%Much work has been done on choosing the correspondence between the vertices of the two cross-sections~\cite{barequet1996piecewise}. 
%Less attention has been paid to the subsequent issue of reconstructing the 3D structure given the correspondence. 
% mention polygon versus graph drawing

Given two simple $n$-vertex polygons, $P=(p_1, \ldots, p_n)$ lying in the $xy$-plane at $z=0$, and $P'=(p'_1, \ldots, p'_n)$ lying in the $xy$-plane at $z=1$, we want to 
interpolate between them by constructing 
a non-self-intersecting \changed{triangulated} surface $\cal S$ homeomorphic to an open-ended cylinder (an annulus), with $P$ at one end and $P'$ at the other end.
\changed{Vertices of $\cal S$ that are not vertices of $P$ or $P'$ are called \emph{Steiner points}.}
We want the surface to be monotone, in the sense that any plane $z=t$ intersects the surface in one simple (non-self-intersecting) polygon. 
%``Simple'' means non-self-intersecting.  \attention{Maybe we should just use ``non-self-intersecting'' throughout.}
Furthermore, we want to maintain the correspondence between $p_i$ and $p'_i$ in 
the following strong sense: for each $i$ there is a path $\pi_i$ of edges in 
\changed{the triangulation of}
$\cal S$ from $p_i$ to $p'_i$, and these paths are vertex disjoint.    
The paths then partition the surface $\cal S$ into interior-disjoint \emph{bands} $B_1, \ldots, B_n$, where $B_i$ is the subset of $\cal S$ between $\pi_i$ and $\pi_{i+1}$.
We call $\cal S$ a \emph{banded surface} 
and we call this problem \emph{banded surface reconstruction between parallel slices} or just ``banded surface reconstruction''.
Figure~\ref{fig:Schonhardt} shows some examples.

\begin{figure}[thb]
    \centering
    \includegraphics[width=\columnwidth]{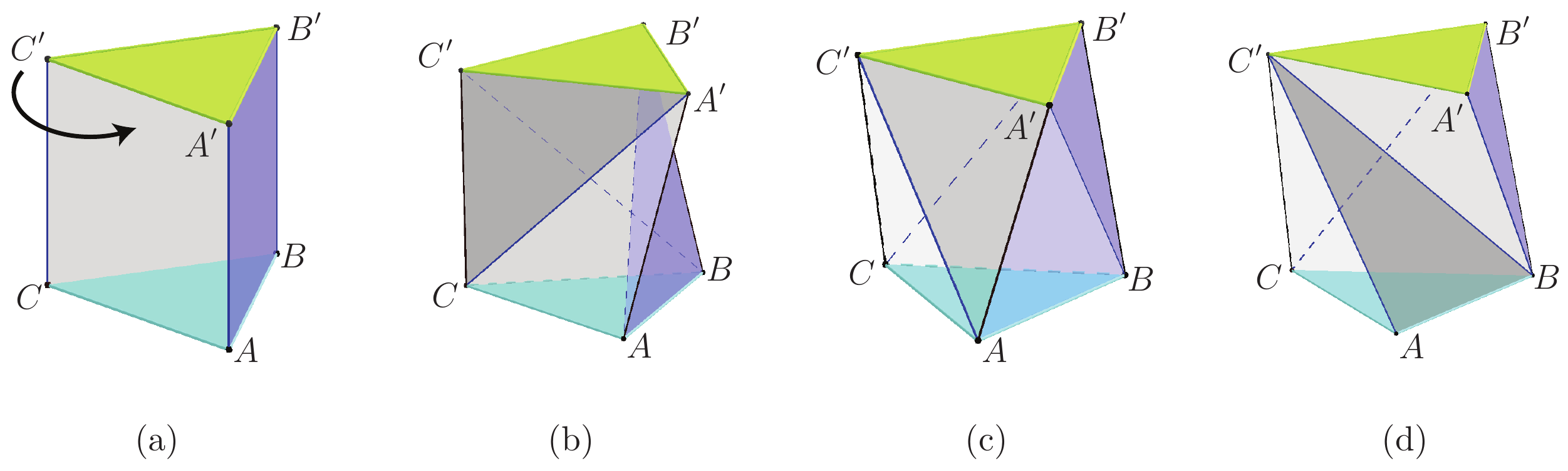}
    \caption{Examples of banded surfaces  without Steiner points for two triangles $P$ and $P'$. (a) To construct $P$ and $P'$, start with a triangular prism based on equilateral triangle $P = ABC$, and then rotate the top triangle to obtain $P' = A'B'C'$. (b) The Sch\"onhardt polyhedron is a banded surface formed by bending each original rectangular face inward to form two triangles, using the ``right'' chords $AB'$, $BC'$, $CA'$.  %It forms a banded surface.
    (c) Using the outward or 
    ``left'' chords, $AC'$, $CB'$, $BA'$ also yields a banded surface (an antiprism when $P'$ is rotated by $60^\circ$).
    (d) An example of a triangulated surface that is not banded due to the lack of a path from $A$ to $A'$ disjoint from $BB'$ and $CC'$.
    }
    \label{fig:Schonhardt}
\end{figure}

\begin{figure}[hbt]
    \centering
    \includegraphics[width=\columnwidth]{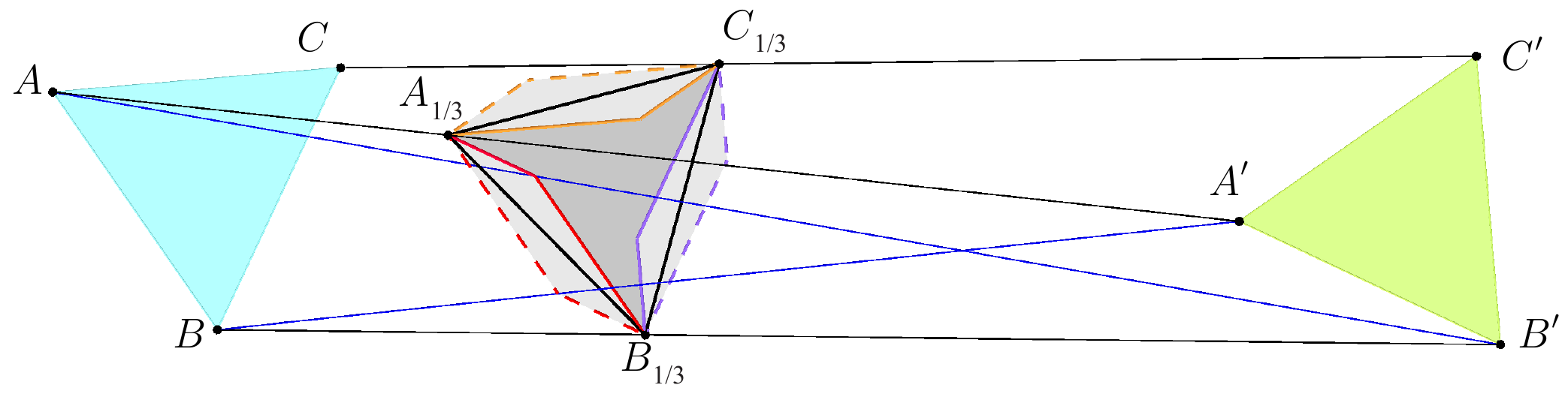}
    \caption{The examples of Figure~\ref{fig:Schonhardt}(b,c) in top-view with triangle $A'B'C'$ translated horizontally.
    (Invariance under translation is proved in Lemma~\ref{lem:translation}.) The cross-section at $z=1/3$ shows the triangle $A_{1/3} B_{1/3} C_{1/3}$ of the linear morph, together with the inward (solid colour) and outward (dashed colour) choices for each edge. Observe that whereas the linear morph uses the edge $A_{1/3}B_{1/3}$ at $z=\frac{1}{3}$, the inward banded surface
    using chord $AB'$ uses two edges (shown in solid red), the first  parallel to $A'B'$ and the second parallel to $AB$, and the outward banded surface using chord $A'B$ uses two edges (shown in dashed red), the first parallel to $AB$ and the second parallel to $A'B'$. 
    }
    \label{fig:Sconhardt-stretch}
\end{figure}

The condition that the surface be homeomorphic to an  annulus prevents 
%Our conditions prevent some 
undesirable ``solutions'' 
such as 
%The condition  that $\cal S$ be homeomorphic to a cylinder prevents us from 
placing one Steiner point $X$  at $z=\frac{1}{2}$ 
and building cones from the configurations at $z=0$ and $z=1$ to $X$. (The fact that these cones do not self-intersect is proved in~\cite{orourke}.)
The condition that the surface is monotone means that the surface provides a morph from $P$ to $P'$, specifically, take $P^t$, for $t \in [0,1]$ to be the intersection of the surface with the plane $z=t$.
Our condition on vertex disjoint paths means that edges of the polygon are maintained throughout this morph in the sense that an edge may become a polygonal path in $P^t$, but it never collapses to a  point.    
%  Note: the following is not a legitimate alteration anyway -- it will self-intersect in general
%Furthermore, the condition about disjoint paths $\pi_i$ prevents us from replacing such a central point $X$  by a 
%\changed{polygon with fewer than $n$ vertices} in the plane $z=\frac{1}{2}$.

%We are also interested in the more general case where $P$ and $P'$ are replaced by planar straight-line drawings of a graph (where both drawings represent the same embedding).
%\comment{Move this comment elsewhere.}

The best bound we know on the number of Steiner points required for a banded surface is $O(n^2)$.
For most of our results 
%In this paper 
we concentrate on the case where no Steiner points are allowed.
Understanding this case may lead to more general solutions where we design $\cal S$ in layers using intermediate polygons (made of Steiner points) at a succession of $z$ values, and build surfaces without additional Steiner points between successive layers.   

When no Steiner points are allowed 
we must use the edges $(p_i, p'_i)$, and our only choice is how to triangulate each quadrilateral $p_i, p_{i+1}, p'_{i+1}, p'_i$.  There are two possible chords for each quadrilateral: the \emph{right} chord $(p_i, p'_{i+1})$ or the \emph{left} chord $(p_{i+1}, p'_i)$. 
The difference between these two choices can be seen in
Figure~\ref{fig:Schonhardt}(b) and (c), and also in Figure~\ref{fig:Sconhardt-stretch}.
An example of two triangles with no banded surface is shown in Figure~\ref{fig:neg-examples}(a).

%Give triangle examples: the Sch\"onhardt polyhedron (see animation by Manuel Garcia Piqueras: 
%\url{https://www.geogebra.org/material/show/id/rmjsjy9g}); a case where a banded surface exists %although the linear morph is not simple; 
%a case where no banded surface exists \url{https://www.geogebra.org/3d/ugny9deb}.}
%Two triangles for which a banded surface/ linear morph does not exist : \url{https://www.geogebra.org/3d/ugny9deb}
%
%For example, the Sch\"onhardt polyhedron is formed from a triangular prism with a twisted top by one choice of chords.  See this animation by Manuel Garcia Piqueras: 
%\url{https://www.geogebra.org/material/show/id/rmjsjy9g} ). 

%%%%%%%%%%%%%%%%%%%%%%%%%%%%%%%%%%%%%
% Results
%%%%%%%%%%%%%%%%%%%%%%%%%%%%%%%%%%%%%
\paragraph{Our Results.} \changed{We prove the following:}
\begin{enumerate}
    \item For $P$ and $P'$ on $n$ vertices, there exists a banded surface with $O(n^2)$ Steiner points.

    \item \changed{There is} %We give
    a polynomial time algorithm (using 2-SAT) to decide the banded surface reconstruction problem when no Steiner points are allowed.
    
    \item \changed{The} %We show that 
    existence of a banded surface without Steiner points is preserved by translating $P'$.
    %\comment{also rotation .  . . .}

    \item %We prove that if
    If $P$ and $P'$ are convex and the linear morph from $P$ to $P'$ preserves \changed{planarity} %convexity
    \changed{(these terms are defined below)}
    then there is a banded surface without Steiner points between $P$ and $P'$.  
    %We give an example to show that this 
    This no longer holds if $P$ and $P'$ are non-convex.
    
    \item In the other direction, %we show that 
    the existence of a banded surface without Steiner points does not imply that the linear morph preserves planarity, not even when $P$ and $P'$ are triangles.  See Figure~\ref{fig:neg-examples}(b).

\end{enumerate}

\begin{figure}[thb]
    \centering
    \includegraphics[width=.7\columnwidth]{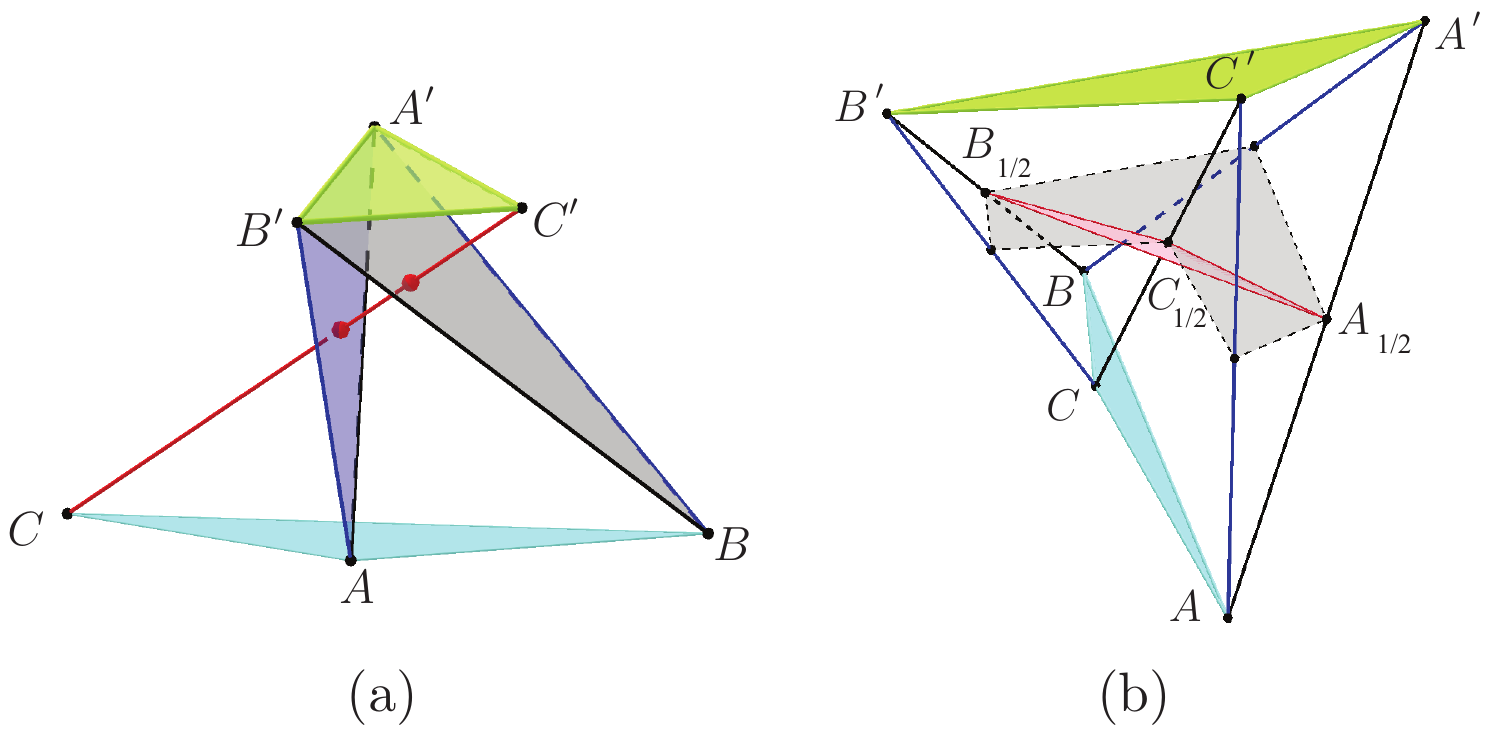}
    \caption{
    (a) Triangles $P  = ABC$ in  the $z=0$ plane and $P' = A'B'C'$ in the $z=1$ plane have no banded surface without Steiner points: the edge $A'B'$ must be in a triangle with either vertex $A$ or vertex $B$ but both those triangles intersect the edge $CC'$.   
    (b) A banded surface between $P = ABC$ and $P' = A'B'C'$ using chords $AC'$, $BA'$, and $CB'$, showing the cross-section (\changed{dashed}, shaded grey) at $z=\frac{1}{2}$.  However, the linear morph from $ABC$ to $A'B'C'$ does not preserve planarity since at $z=\frac{1}{2}$ the triangle $A_{1/2} C_{1/2} B_{1/2}$ (shown in red) is inverted.
    }
    \label{fig:neg-examples}
\end{figure}

%%%%%%%%%%%%%%%%%%%%%%%%%%%%%%%%%%%%%
% Previous Work
%%%%%%%%%%%%%%%%%%%%%%%%%%%%%%%%%%%%%
\paragraph{Previous Work.}
Gitlin, O'Rourke and Subramanian \cite{orourke} considered a similar problem of joining two polygons via a triangulated surface without adding Steiner points.
However, they did not 
\changed{require disjoint paths}
%impose our constraint that there should be disjoint paths 
from $p_i$ to $p'_i$, %.  This 
which gives a lot more freedom, \changed{e.g., the two polygons can have different numbers of vertices.}  
Essentially, every edge of $P$ must be in a triangle with some vertex of $P'$, and vice versa, and these triangles must 
%be internally disjoint. 
form a non-self-intersecting surface homeomorphic to an annulus.
%In fact, in their scenario, the two polygons can have different numbers of vertices.
Their main result was a construction of a pair of polygons on 63 vertices with no triangulated surface between them.  Their proof involved a computer search.   
Barequet and Steiner~\cite{barequet2008matability} gave a slightly 
simpler example on 45 vertices.
The problem of testing whether two polygons can be joined via a triangulated surface without Steiner points is not known to be NP-complete (nor in P).   
There is a surprising  upper bound on  the  number of Steiner  points required for a triangulated surface.
Geiger~\cite[Appendix A]{geiger1993} proved that  
it suffices to add two Steiner points, one on an edge of $P$ and one on an edge of $P'$.  To do this, he first constructed a degenerate surface consisting of two cones, one with $P$ as a base 
and  the rightmost vertex of $P'$ as its apex, and one with $P'$ as a base, and the leftmost vertex of $P$ as its apex.
These two cones share one edge, but by adding the two Steiner points, the shared edge can be pulled apart so that the  two cones become a single surface homeomorphic to an annulus. 
This construction is at the heart of our argument in Section~\ref{sec:construct} that $O(n^2)$ Steiner points suffice to construct a  banded surface.
%It is an open question whether this version of the problem is NP-complete.

In more applied work, there is a vast literature about interpolating between two families of nested polygons 
lying in parallel planes via a %piecewise-linear 
triangulated surface, see~\cite{barequet2004contour,barequet1996piecewise}. 
%\comment{No nice bounds on Steiner points.} 
Barequet and Sharir~\cite{barequet1996piecewise} write: 
%``Practically all currently-known methods for surface reconstruction from parallel slices take the matability of the polygons for granted. That is, they 
%assume a priori the existence of a non-self-intersecting triangulated surface defined %over the vertices of the two polygons, and connecting them. 
``The primary concern in the literature has usually been to find fast heuristics for selecting a `good' reconstruction among the many available solutions.''
There is little work analyzing the number of Steiner points, or examining when a solution with no Steiner points is possible.

Our problem is related to the problem of finding a piecewise linear embedding of a 2D simplicial complex 
in 3D,  which  was  recently shown to be NP-hard~\cite{demesmey2018embeddability}.
(One dimension down this is easy, since it is the 
%\comment{well-studied?} Anna  to Pavle: There is a linear time algorithm -- that's what I meant by "well-solved", but I changed to say "easy"
problem of finding a (poly-line) planar drawing of a graph.)
%(By constrast, a piecewise linear embedding of a 1D simplicial complex in 2D is a (poly-line) planar graph drawing.)
Specifically, the 2D complex that we want to embed in 3D 
%Our 2D complex 
consists of the quadrilaterals $p_i, p_{i+1},p'_{i+1}, p'_{i}$, and we have the further constraint that the embeddings of $P$ and $P'$ are already fixed in the 3D space.  
Our additional structure ensures that there always is a solution % using $O(n^2)$ Steiner points, 
so  the interesting  problems are to minimize the number of Steiner points and/or to optimize other parameters of the solutions such as the bit  complexity of the Steiner points, the lengths of the paths from $p_i$ to $p'_i$, or etc.

%\attention{Why a banded surface exists.  Best bound on number of Steiner points.}

%%%%%%%%%%%%%%%%%%%%%%%%%%%%%%%%%%%%%
% Relationship to Morphing
%%%%%%%%%%%%%%%%%%%%%%%%%%%%%%%%%%%%%
\paragraph{Relationship to Morphing.}
A morph is a continuous transformation from one shape to another.  In particular, 
a \emph{morph} from an initial simple polygon  [or planar straight-line graph drawing] $P^0$ to a final one, $P^1$, with the same labelled vertices,
is a continuously changing family of polygons [or graph drawings] $P^t$ indexed by  time $t \in [0,1]$.
A morph \emph{preserves planarity} if all intermediate polygons [drawings] $P^t$ are planar.
%\attention{Fix to ``simple''?}
In a \emph{linear morph}  each vertex moves on a straight line from its initial position to its final position at constant speed (where the speed of a vertex depends on the distance it must travel), and 
an edge is always drawn as a line segment between its endpoints.

Our problem of reconstructing a 3D polyhedral structure between two planar drawings
%(under the name ``3D shape reconstruction'') 
is closely related to morphing---the $z$ direction corresponds to time $t$ in the morph.  
In fact, it is claimed (for example, by Surazhsky and Gotsman~\cite{surazhsky2004}) that morphing algorithms %provide a solution to 
solve 3D shape reconstruction.  
We now examine this claim more closely.
Figure~\ref{fig:morphing-and-surfaces} 
illustrates the idea. Initialize $P^0$ to $P$ and $P^1$ to $P'$.   
Given a morph $P^t$, $t \in [0,1]$ between $P^0$ and $P^1$, take a finite set of ``snapshots'' at time points $t_1, \ldots, t_k$, and form a quadrilateral ``patch'' between successive vertices $p_i$ and $p_{i+1}$ at times $t_j$ and $t_{j+1}$.   Each patch is a ruled surface, and the union of the patches provides a surface in 3D joining $P^0$ and $P^1$.
In order to obtain a piece-wise linear surface we must
replace each quadrilateral patch by two triangles.   
This may cause the surface to self-intersect (if it doesn't already).
%There are two disadvantages of this approach---the surface is not  piece-wise linear  and it may  self-intersect. 
%shows how a morph between two polygons yields a 3D surface composed of quadrilateral \changed{``patches'', each a ruled surface}.  %hyperbolic paraboloid}.
%Even if the linear morph from one snapshot to the next preserves planarity (which Surazhsky and Gotsman did not guarantee), 
%It is not obvious how to do this---replacing quadrilaterals by pairs of triangles may 
%cause the surface to self-intersect. 
%(In fact, if the morph is not linear, then the quadrilateral \changed{patches}  themselves will intersect.)
%(Figure~\ref{fig:star} shows such an example.)
It seems intuitive that self-intersections can be avoided by 
taking sufficiently many snapshots, but such analysis is lacking.

\begin{figure}
    \centering
    \includegraphics[width=0.4\columnwidth]{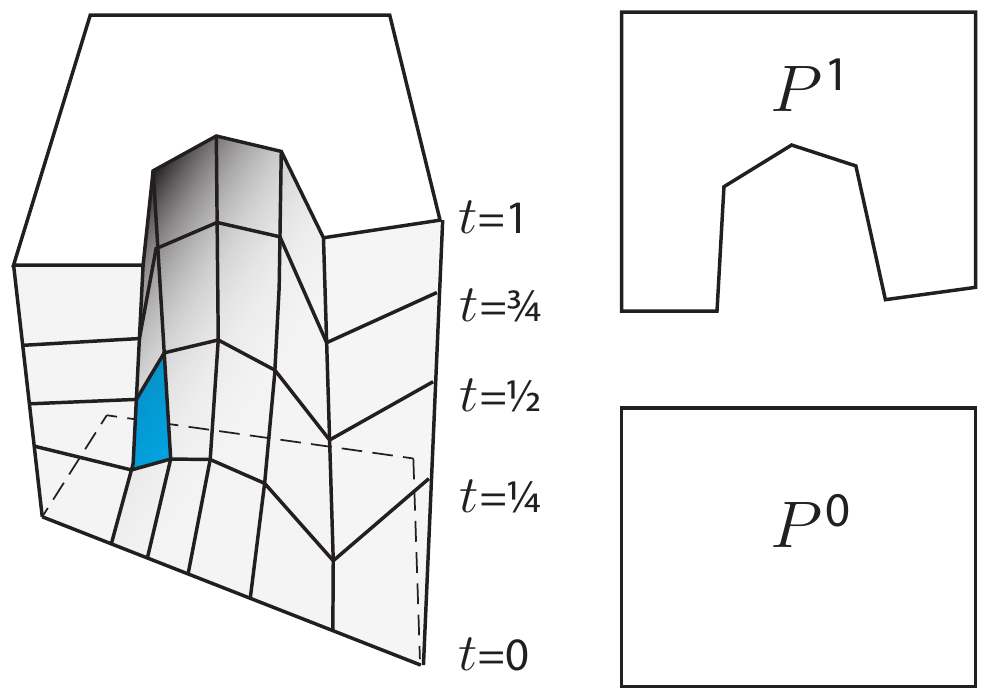}
    \caption{A morph from a rectangle $P^0$ at time (or $z$-coordinate) $t=0$ to the ``arch-shaped'' polygon $P^1$ at time $t=1$ yields a 3D interpolation by taking ``snapshots'' of the morph at intermediate time points $t = \frac{1}{4}, \frac{1}{2}, \frac{3}{4}$, 
    and joining corresponding vertices between one snapshot and the next.
    Note that the resulting quadrilateral patches (one of which is coloured blue) are not planar in general.
    This figure is loosely based on one by Surazhsky and Gotsman~\cite[Figure 10]{surazhsky2004}.}  
%    How a morph between two planar graph drawings (the ``top'' and ``bottom'' slices) yields some kind of 3D interpolation between them.  Note that the quadrilaterals on the surface in 3D (one of which is coloured blue) are not planar in general.}
    \label{fig:morphing-and-surfaces}
\end{figure}

%Across the top of the figure are some ``snapshots in time'' during a morph from the initial to the final
%planar graph drawing.  The main part of the figure shows the 3D surface formed by placing these snapshots at successive $z$ levels using Steiner points.  

An algorithm by 
Alamdari et al.~\cite{alamdari2016morph}
finds ``piece-wise linear'' morphs that consist of a %linear number 
sequence of planarity-preserving linear morphs.
%morphs one planar straight-line 
%drawing of an $n$-vertex graph $G$ to another planar straight-line drawing of $G$
%via a planarity-preserving morph 
%consisting of a succession of $O(n)$ linear morphs (a ``piece-wise linear'' morph).
This would provide a solution to banded surface reconstruction if we could 
show how to add Steiner
points to turn each linear morph into a triangulated surface. 

In the other direction, a banded surface (even one with Steiner points)  can be interpreted as a morph between the polygons $P$ and $P'$, albeit a morph in  which each edge may become a poly-line.
Such ``morphs with bent edges'' have been investigated~\cite{lubiw-petrick} and come with small grid guarantees, unlike the piece-wise linear morphs of~\cite{alamdari2016morph}. 
A banded surface without Steiner points provides a morph with the interesting property that  
in any intermediate drawing of the morph, an edge $e$ appears as a path of two line segments, one in the direction of the initial version of $e$ and the other in the direction of the final version of $e$.
See Figure~\ref{fig:Sconhardt-stretch}.
Such morphs may be valuable for visualizations. 
We note that there is work on morphing while maintaining edge directions~\cite{biedl2005morphing}---this only applies in the restricted situation where the initial and final polygons have corresponding edges with the same directions.
%\comment{Hamide's thesis.}

%We suggest a different approach.  When a banded surface without Steiner points is interpreted as a morph, edges in 
%intermediate snapshots appear with one bend each.  This offers more freedom than a straight-line morph, %and we hope to use this freedom in creating the banded surface. 
%
%Note: There is work on morphing while maintaining edge directions.  This applies only in the restricted situation where the initial and final graphs have corresponding edges with the same directions.  In more general situations 
%the morph that results from a banded surface without Steiner points gives a slightly similar result.  Consider any initial edge $e$ with direction $d$ corresponding to final edge $e'$ with direction $d'$.  In any intermediate drawing in the morph resulting from a banded surface without Steiner points the corresponding ``edge'' is a path of two line segments, one in direction $d$ and one in direction $d'$. In other words, the morph does not introduce any new directions of line segments.  

To summarize, it seems worth investigating to what extent linear morphs provide banded surfaces, and to what extent banded surfaces provide morphs.

%%%%%%%%%%%%%%%%%%%%%%%%%%%%%%%%%%%%%
% Finding a banded surface without Steiner points 
%%%%%%%%%%%%%%%%%%%%%%%%%%%%%%%%%%%%%
\section{Finding a Banded Surface with/without Steiner Points}
\label{sec:construct}

In this section we show that $O(n^2)$ Steiner points suffice to construct a banded surface between $n$-vertex polygons $P$ and $P'$, and  we give an algorithm %using  2-SAT 
to find---if it exists---a banded surface without Steiner points between $P$ and $P'$.

As discussed above, Geiger~\cite{geiger1993} showed how to add two Steiner points to construct a triangulated surface between $P$ and $P'$.  The surface he constructs is monotone, homeomorphic to an  annulus, and consists of $O(n)$ triangles.  To construct a banded surface, take disjoint paths on the surface from $p_i$ to $p'_i$ for each $i$, and  refine the triangulation to include all the line segments  of these paths.  Since there are $n$ paths each crossing $O(n)$ triangles, the result is a surface of $O(n^2)$ triangles. 
We note that the same bound $O(n^2)$ can be obtained using a technique from Piecewise Linear (PL) Topology 
in which ears of the polygon are collapsed in successive steps.  (See ``elementary contractions'' in the classical book~\cite{Seifert-Threlfall} or the lecture notes~\cite[p.~30]{Lackenby}.) 
One ear collapse replaces a convex vertex $p_i$ by a vertex on the line segment from $p_{i-1}$ to $p_{i+1}$, resulting in a polygon with one fewer line segments. 
The new polygon is placed on a slightly higher $z$ plane and the two successive polygons are joined with $n$ triangles.
In  this way we collapse $P$ (working upwards in 3D) to a triangle $T$, collapse $P'$ (working downwards in 3D) to a triangle $T'$ and finally build a surface joining $T$ and $T'$. 
This may produce a nicer surface than the one obtained from Geiger's construction.

We next describe an algorithm using 2SAT to find a banded surface without Steiner points, if one exists.
%In this section we give an algorithm using  2-SAT to find---if it exists---a banded surface without Steiner points between
%two $n$-vertex polygons, $P=(p_1, \ldots, p_n)$ lying in the $xy$-plane at $z=0$, and $P'=(p'_1, \ldots, p'_n)$ lying in the $xy$-plane at $z=1$. 
The edges $p_i p'_{i}$ must be used.  
For each $i=1, \ldots, n$ we have the choice of the right chord $p_i p'_{i+1}$ or the left chord $p_{i+1} p'_i$.  
%Relative to our convention of ordering the polygon vertices counterclockwise, we call $p_i p'_{i+1}$ the ``right'' chord and $p_{i+1} p'_i$ the ``left'' chord.  
Let the Boolean variable $R_i$ be 1 if the right chord is chosen and 0 otherwise.  
Each chord choice determines two triangles of the surface, for example
$R_i=1$ determines triangles $p_i p_{i+1} p'_{i+1}$ and $p_i p'_{i+1} p'_i$.
%, and similarly  for $R_i=0$.
We say that chord choices for $i$ and $j$ \emph{conflict} if the resulting open triangles intersect.
Note that this can be tested, for given $i,j$, in constant time.
The problem of choosing chords to form a non-self-intersecting surface can be formulated as a Boolean satisfiability problem by adding a clause to prohibit conflicts,  e.g., if chord choices $R_i$ and $\neg R_j$ conflict then we add the clause $\neg(R_i \wedge \neg R_j)$.
\changed{Note that there are $O(n^2)$ clauses.}

There is a banded surface without Steiner points if and only if the resulting clauses are satisfiable.  Because all clauses have two variables, the result is a 2-SAT instance.   Since 2-SAT can be solved in linear time~\cite{aspvall1979linear}, we have:

\begin{lemma}
\label{lem:2SAT}
There is a \changed{quadratic} %linear 
time algorithm that either finds a banded surface without Steiner points, or declares that no such surface exists.
\end{lemma}

% ran out of time for this:
%Figure~\ref{fig:mixed-chords} shows an instance where the only banded surface involves some right chords and some left chords. \attention{Make figure - like  the star figure later on, but just a non-convex quadrilateral.}

%%%%%%%%%%%%%%%%%%%%%%%%%%%%%%%%%%%%%
% Some Conditions
%%%%%%%%%%%%%%%%%%%%%%%%%%%%%%%%%%%%%
\section{Conditions for Existence of 
a Banded Surface without Steiner Points}

One approach to banded surface reconstruction with Steiner points is 
to subdivide the interval $z \in [0,1]$ into $0{=}z_0, z_1, \ldots, z_k{=}1$ and place an $n$-vertex polygon at each $z_i, 0 < i < k$ so that each successive pair of polygons admits a banded surface without Steiner points.  Using this approach, the final solution would have $nk$ Steiner points.
%\comment{Maybe mention rotation in particular.}

In order to design the intermediate polygons,
it would be good to have conditions for when two polygons admit a banded surface without Steiner points.  (Our polynomial-time test from the  previous section  does not seem helpful when the polygons are not given).

In this section we explore two situations where we can guarantee the existence of a banded surface without Steiner points.
We show:
\begin{enumerate}
    \item Translation of $P'$ in the $z=1$ plane preserves the existence  of a banded surface  without Steiner points
    \changed{(Lemma~\ref{lem:translation}).}
%    linear morph and banded surface. \comment{Write better.}
%    In particular, if $P'$ is a translation of $P$ then a banded surface exists.
    \item If $P$ is convex and $P'$ is a rotation of $P$ by an angle less than \changed{$\pi$}, %$180^\circ$ 
    then a banded surface without Steiner points exists
    \changed{(Lemma~\ref{lem:rotation}).}
    The example of Figure~\ref{fig:star} shows that this property does not hold more generally, not even for a star-shaped polygon.
\end{enumerate}

We first show 
%some general results on 
how translation of the target-polygon affects the intermediate polygons in a linear morph:

\begin{lemma} 
\label{lem:translation_new}
Let $P$ be an $n$-vertex polygon in the $z=0$ plane and $P'$ be an $n$-vertex polygon in the $z=1$ plane.  
%Let $Q=P$ and 
Let $Q'$ be a translation of $P'$ within the $z=1$ plane.  
For any $0<t<1$, if $P_t$  is the polygon at time $t$ during \changed{the} linear morph from $P$ to $P'$, and $Q_t$ is the polygon at time $t$ during \changed{the} linear morph from $P$ to $Q'$, then $Q_t$ is a translation of $P_t$ within the $z=t$ plane.
\end{lemma}
\begin{proof}
Set $s=Q'-P'$ to be the translation vector and consider an arbitrary point $p$ of $P$ that morphs to point $p'$ of $P'$ and $q'$ of $Q'$.  We have $q'=p'+s$, and hence
$$q_t = (1-t)p+tq' = (1-t)p+tp' + ts = p_t + t\cdot s$$
so polygon $Q_t$ is a translation of $P_t$ by $t\cdot s$.
\end{proof}

In particular, if the linear morph from $P$ to $P'$ preserves 
\changed{planarity}, %simplicity 
%and/or convexity, 
then the same holds for the linear morph from $P$ to any translation of $P'$.  We can argue the same for banded surfaces:

\begin{lemma} 
\label{lem:translation}
Assume that $P, P'$ and $Q'$ are as in Lemma~\ref{lem:translation_new}.  
If there is a banded surface without Steiner points between $P$  and $P'$,  then 
the same choice of chords yields
%there is 
a banded surface without Steiner points between $P$ and $Q'$.
\end{lemma}
\begin{proof}
\changed{We show that the banded surface between $P$ and $P'$ is the same as the linear morph between two modified polygons $P_D$ and $P'_D$, which we now define.}
%We need an intermediary result.  Define new polygons $P_D$ and $P_D'$ as follows.  
Initially start with $P$ and $P'$. For each $i=1,\dots,n$, if we \changed{chose} the right chord $p_ip_{i+1}'$, then duplicate vertex $p_i$ in $P_D$ (inserting an edge of length 0) and duplicate vertex $p_{i+1}'$ in $P_D'$.  Proceed symmetrically if we \changed{chose} the left chord.  Now consider \changed{the} linear morph from $P_D$ to $P_D'$, where vertices that have been inserted due to a chord correspond to each other.  Say we \changed{chose} the right chord $p_ip_{i+1}'$.  Then the zero-length edge $p_ip_i$ in $P_D$ morphs to edge $p_i'p_{i+1}'$ in $P_D'$, hence forms a triangle.  Likewise edge $p_ip_{i+1}$ in $P_D$ morphs to zero-length edge $p_{i+1}'p_{i+1}'$ in $P_D'$, and also forms a triangle.  The two triangles together form exactly the part of the banded surface between edges $p_ip_{i+1}$ and $p_i'p_{i+1}'$ in $P$ and $P'$.

\changed{Since the banded surface is the same as the linear morph from $P_D$ to $P'_D$} the
%Hence the banded surface between $P$ and $P'$ is the same as the linear morph between two (suitably modified) polygons $P_D$ and $P_D'$.  The 
result now follows from Lemma~\ref{lem:translation_new}.
\end{proof}

\changed{We now turn to rotations, beginning with this result on} 
%We need a second helper-result that studies 
linear morphs when the target-polygon is a rotation of the source-polygon:

\begin{lemma}
\label{lem:linear-rotate}
Let $P$ be a polygon and let $P'$ be a rotation of $P$ about an origin $X$ by an angle $\alpha$.  For any $0<t<1$ 
let $P_t$ be the polygon at time $t$ during a\ linear morph from $P$ to $P'$.  If $\alpha\neq \pi$ or $t\neq \frac{1}{2}$
then $P_t$ is a rotated copy of $P$ that has been scaled by $s_t\neq 0$.
\end{lemma}

\begin{proof}
We consider $P$, $P'$ and $P_t$ projected to the $xy$ plane.
If $\alpha = \pi$ then every point $p$ of $P$ maps to to $-p$ in $P'$, which implies $p_t=(1-2t)p$.  So $P_t=s_t P$ for $s_t=1-2t$, which is non-zero for $t\neq \frac{1}{2}$.    

Now suppose that $\alpha <  \pi$ (the case $\alpha>\pi$ is symmetric).  For any point $p$ of $P$, consider the triangle $\Delta_p:=\Delta pXp'$, where $X$ is the center of the rotation.  
Note that $\Delta_p$ and $\Delta_q$ are similar for any two points $p$ and $q$ of $P$, since they both have angle $\alpha$ and two equal-length incident sides; in particular $\Delta_q$ is obtained from $\Delta_p$ by scaling by $||q||/||p||$ and (possibly) rotating.  Also notice that $p_t$
travels along the side of $\Delta_p$ opposite to angle $\alpha$, and is at the point that divides the side at ratio $t/(1-t)$.  We can view $p_t$ as having been rotated by some angle $\theta_t$ and scaled by some $s_t>0$. Both $\theta_t$ and $s_t$ are independent of the choice of $p$ since all triangles $\Delta_p$ are similar.  Therefore $P_t$ is obtained from $P$ by scaling by $s_t$ and rotating by $\theta_t$.
%\comment{Therese's proof has more details. See main-save.tex.}
\end{proof}
%\comment{space permitting add a picture of $p$ travelling to $p'$ along $\Delta_p$.}  Anna: we're out of space (and time)
\begin{lemma}
\label{lem:rotation}
Let $P$ be a convex polygon and let $P'$ be a rotation of $P$ about an origin $X$ by an angle $\alpha< \pi$.  %180^\circ$.  
Then there is a banded surface without Steiner points between $P$ and $P'$.
\end{lemma}
\begin{proof}
Observe first that 
the linear morph from $P$ to $P'$ preserves %both 
\changed{planarity} %simplicity 
%and convexity, 
since, \changed{by Lemma~\ref{lem:linear-rotate},} each intermediate polygon is a rotated and scaled copy of $P$.  
By Theorem~\ref{theorem:convex-morph-revised} (forthcoming, but there is no circularity) this implies the existence of a banded surface without Steiner  points.
\end{proof}

%%%%%%%%%%%%%%%%%%%%%%%%%%%%%%%%%%%%%
% Relationship to Linear Morph 
%%%%%%%%%%%%%%%%%%%%%%%%%%%%%%%%%%%%%
\section{Linear Morphing versus Banded Surface Reconstruction}

%Anna: I moved this section to the end since the main lemma will use the translation property

In this section we compare the existence of a planarity-preserving linear morph from $P$ to $P'$ and the existence of a banded surface without Steiner points.  
In general, these two properties are independent, i.e., neither implies the other.
Figure~\ref{fig:neg-examples}(b) shows an example of two triangles that have a banded surface  without Steiner points, but the linear morph does not preserve planarity.  
Figure~\ref{fig:star} shows an example of two stars that do not have a banded surface without Steiner points, but the linear morph preserves planarity.

%\attention{This  is the new proof.}
When the \changed{polygons $P$ and $P'$ are convex},
%polygon is convex throughout the linear morph, 
there is an implication:

\begin{theorem}
\label{theorem:convex-morph-revised}
If $P$ and $P'$ are convex and the linear morph from $P$ to $P'$ preserves 
\changed{planarity}, %convexity, 
then there is a banded surface without Steiner points between $P$ and $P'$.
\end{theorem}
\begin{proof}
Let $p_i^t$ be the position of the $i$th vertex at time ($z$-coordinate) $t$ during the linear morph. In  particular, $p^0_i = p_i$ and $p^1_i = p'_{i}$. 
Let $P^t$ be the polygon at time $t$ during the morph.
Note that $P^t$ is not necessarily convex.
By our convention of numbering polygons in counterclockwise order, the inside of $P$ is to the left of $p_i p_{i+1}$, and the inside of $P'$ is to the left of $p'_i p'_{i+1}$. 
Also, because the linear morph preserves planarity, the inside of $P^t$ is to the left of $p_i^t p_{i+1}^t$.

We begin by defining the surface ${\cal S}$, i.e., which chords to use. 
Let $v^0_i$ be the vector $p_{i+1} - p_i$ 
%projected to 
\changed{in} the $xy$ plane, and let 
$v^1_i$ be the vector $p'_{i+1} - p'_i$ projected to the $xy$ plane.
Let $\theta_i$ be the angle between $v^0_i$ and $v^1_i$, measured towards the inside of $P$,
\changed{as shown in Figure~\ref{fig:convex-morph}}.   We distinguish 3 cases:
\begin{itemize}
    \item If 
    %$\theta_i$ is convex, i.e., 
    $\theta_i < \pi$, use the left  chord $p_{i+1}p'_i$. 
    In the cross-section of $\cal S$ at $z$-coordinate (or time) $t$, the edge $p_i^t p_{i+1}^t$ is replaced by a segment in the direction $v^0_i$ followed by a segment in the direction $v^1_i$.  
    We call the resulting triangle $\Delta^t_i$ and refer to it as a 01 triangle. 
    Observe that $\Delta^t_i$ lies to the outside of the edge $p_i^t p_{i+1}^t$.  
    See Figure~\ref{fig:convex-morph}(a).
    \item If 
    %$\theta_i$ is reflex, i.e., 
    $\theta_i > \pi$, use the right chord $p_ip'_{i+1}$.  
    Then, in the cross-section at $z$-coordinate $t$, the edge $p_i^t p_{i+1}^t$ is replaced by a segment in the direction $v^1_i$ followed by a segment in the direction $v^0_i$. 
    We refer to the  resulting triangle $\Delta^t_i$ as a 10 triangle. 
    Again, $\Delta^t_i$ lies to the outside of $p_i^t p_{i+1}^t$, see Figure~\ref{fig:convex-morph}(b).
    \item If $\theta_i = \pi$, use either chord---in this case the quadrilateral $p_i, p_{i+1} p'_{i+1} p'_{i}$ is coplanar, and $\Delta_i^t$ collapses to  the edge $p_i^t p_{i+1}^t$.
\end{itemize}

\begin{figure}[htb]
    \centering
    \includegraphics[width=.6\columnwidth]{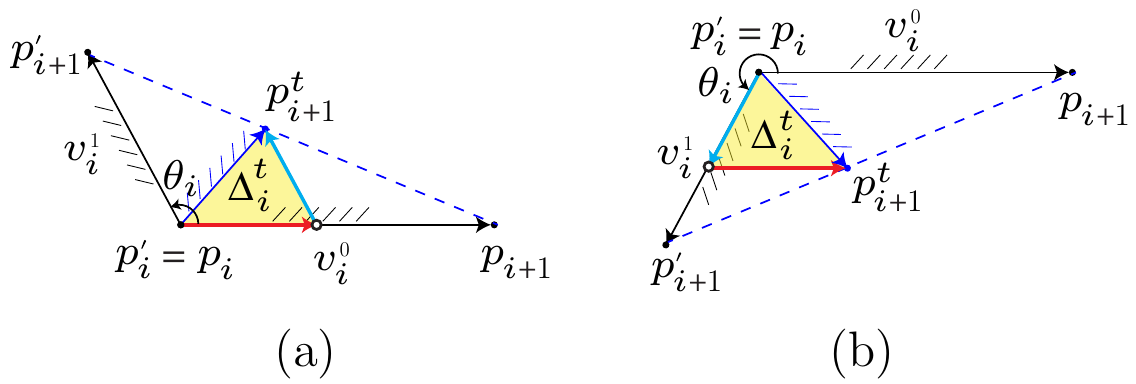}
    \caption{ A top view to illustrate choosing chords in the proof of Theorem~\ref{theorem:convex-morph-revised}.  In order to show the angles clearly, $P'$ has been translated so that $p_i$ is at the same $xy$-coordinates as $p'_{i}$. (Lemma~\ref{lem:translation} justifies this.)  Note that $p_i^t$ then remains at these same $xy$-coordinates.
    Hatching indicates the inside of the polygon
    \changed{on that edge}.  
    (a) If $\theta_i < \pi$, use a left chord to obtain a 01 triangle $\Delta^t_i$. 
    (b) If $\theta_i > \pi$, use a right chord to obtain a 10 triangle $\Delta^t_i$.
    The segments that replace $p_i^t p_{i+1}^t$ are shown in red/cyan.
    }
    \label{fig:convex-morph}
\end{figure}

We now prove that ${\cal S}$, as defined by the above  chord choices, is non-self-intersecting, which proves that $\cal S$ is a banded surface without Steiner points. In  particular, we will prove that ${\cal S}^t$, the cross-section of $\cal S$ at $z$-coordinate $t$ is a simple polygon.
By assumption, the polygon $P^t$ with vertices $p^t_1, p^t_2, \ldots, p^t_n$ is \changed{simple}.  %convex.
${\cal S}^t$ consists of $P^t$ plus triangles $\Delta^t_i$ added to the outside of each edge.  See Figure~\ref{fig:convex-morph-config-2}.
We will show that no two triangles intersect.

\begin{claim}
\label{claim:0101}
Suppose that $\Delta_i^t, \Delta_{i+1}^t, \ldots, \Delta_j^t$ are all 01 triangles. Let $r_i^0$ be the ray  from $p_i^t$ in  the direction $v^0_i$.  
Then none of these triangles cross $r_i^0$ from its left to its right.
\end{claim}
\begin{proof}
It suffices to prove that no triangle crosses the ray of the previous triangle, so consider triangle $\Delta_{i+1}^t$ and $r_i^0$.  The apex of $\Delta_{i+1}^t$ lies on $r_{i+1}^0$.  Rays $r_i^0$ and $r_{i+1}^0$ 
\changed{emanate from the endpoints of the edge $p_i^t p_{i+1}^t$}
%both leave the convex polygon $P^t$ 
and the angle between $r_i^0$ and $r_{i+1}^0$ is positive (counterclockwise). 
Thus the apex of $\Delta_{i+1}^t$ lies to the left of $r_i^0$.
%This is because the rays $v^0_i, v^0_{i+1}, \ldots, v^0_j$ progress counterclockwise, see Figure~\ref{fig:convex-morph-config-2}.
%\comment{Maybe say more.}
\end{proof}

\begin{figure}[htb]
    \centering
    \includegraphics[width=.4\columnwidth]{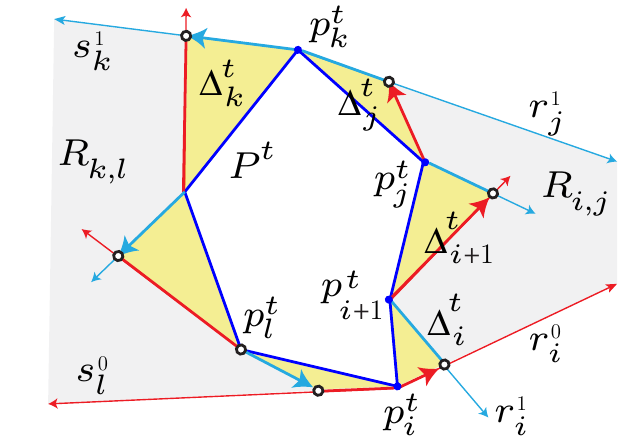}
    \caption{Polygon $P^t$ \changed{(in blue)} with 01 triangles $\Delta_i^t, \ldots, \Delta_j^t$ and 10 triangles \changed{$\Delta_{k}^t, \ldots , \Delta_l^t$}.  In general, there may be many alternations between 01 and 10 triangles.  
    }
    \label{fig:convex-morph-config-2}
\end{figure}

\changed{
Define $r_i^1$ to be the ray from $p_{i+1}^t$ in the direction $- v^1_i$.
Thus a 01 triangle $\Delta_i^t$ is bounded by $r_i^0$ and $r_i^1$. 
Symmetrically, for a 10 triangle, define $s_i^1$ to be the ray from $p_i^t$ in the direction $v^1_i$, 
and $s_i^0$ to be the ray from $p_{i+1}^t$ in the direction $- v^0_i$.
Thus a 10 triangle $\Delta_i^t$ is bounded by $s_i^1$ and $s_i^0$. 
}

\changed{
From Claim~\ref{claim:0101}, by symmetry, we obtain (see Figure~\ref{fig:convex-morph-config-2}):
}
\begin{claim}
\label{claim:1010}
\changed{If $\Delta_i^t, \Delta_{i+1}^t, \ldots, \Delta_j^t$ are 01 triangles then none of them cross 
$r_j^1$ from right to left.
If $\Delta_k^t, \Delta_{k+1}^t, \ldots, \Delta_l^t$ are 10 triangles then none of them cross $s_k^1$ from left to right and none of them cross $s_l^0$ from right to left.
}
%
%Let $r_i^0$ be the ray  from $p_i^t$ in  the direction $v^0_i$.  
%Then none of these triangles cross $r_i^0$ from its left to its right.
%Suppose that \changed{$\Delta_k^t, \Delta_{k+1}^t, \ldots, \Delta_l^t$} are all 10 triangles. Let 
%\changed{$s_k^t$ be the ray with head at $p_{k}^t$ in  the direction $v^1_k$.}  
%Then none of these triangles cross \changed{$s_k^t$} from its left to its right.
\end{claim}

These two claims imply that ${\cal S}^t$ is simple if all the triangles are the same (all 01 or all 10). It remains to consider the possibility that there are triangles of both types.

\begin{claim}
\label{claim:101}
Suppose $\Delta_{i-1}^t$ is a 10 triangle and $\Delta_i^t$ is a 01 triangle.  
Then $\Delta_{i-1}^t$ and $\Delta_i^t$ are disjoint.
\changed{Furthermore, $P^t$ is convex at $p_i^t$.}
\end{claim}
\begin{proof}
We analyze the top-view projection with $P'$ translated so that $p_i$ and $p'_i$ are at the same $xy$-coordinates.

\changed{
Consider the angle $\alpha_i^t = \angle p_{i+1} p_i p_{i+1}^t$.  Because $\Delta_i^t$ is a 01 triangle, $\alpha_i^t$ goes from 0 to $\theta_i < \pi$.
Similarly, because $\Delta_{i-1}^t$ is a 10 triangle, the angle $\alpha_{i-1}^t = \angle p_{i-1} p_i p_{i-1}^t$ goes from 0 to $2\pi - \theta_{i-1} < \pi$.
%$\theta_{i-1} - \pi < \pi$.
%(Note that Figure~\ref{fig:convex-morph}(b) draws the 10 triangle at reference point $p_{i-1} = p'_{i-1}$ so the frame of reference must be changed to get $\alpha_{i-1}^t$.)

If $p_i p^t_{i-1}$ and $p_i p^t_{i+1}$
%these two \attention{(say what ``these'' are)}
cross over each other, as in Figure~\ref{fig:101}(a), i.e., 
$\theta_i + 2\pi - \theta_{i-1} \ge \angle p_{i-1} p_i p_{i+1}$,
then there must be some time $t$ when $\alpha_i^t + \alpha_{i-1}^t = \angle p_{i-1} p_i p_{i+1}$, i.e., angle $\angle p^t_{i-1} p^t_i p^t_{i+1}$ becomes 0.  But we assumed that $P^t$ remains simple, so this cannot happen.

\begin{figure}[htb]
    \centering
    \includegraphics[width=.6\columnwidth]{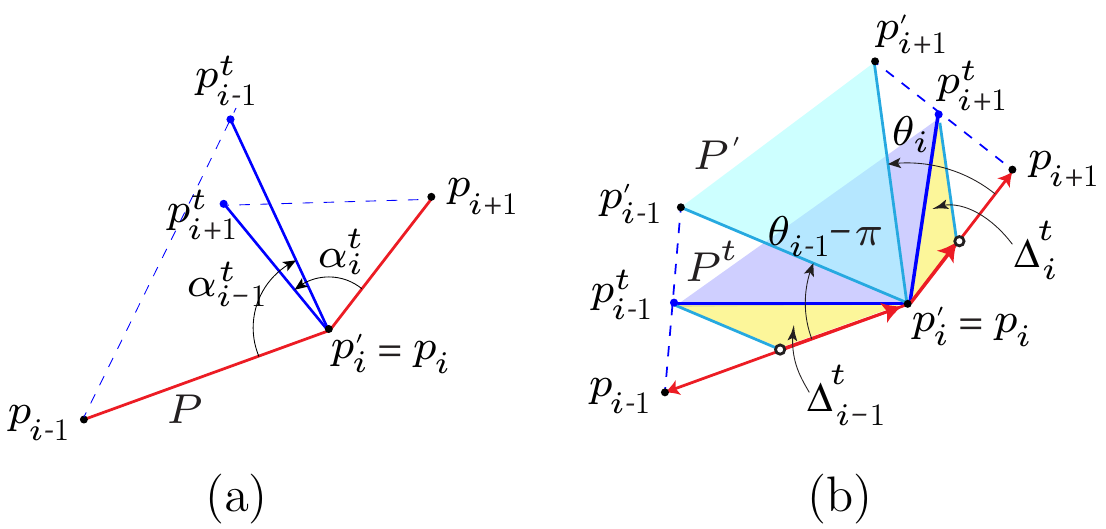}
    \caption{Illustration for the proof of Claim~\ref{claim:101}: (a) top view projection at $p_i = p'_i$ showing the angles $\alpha^t_i$ and $\alpha^t_{i-1}$; 
    (b) because $\theta_i + 2\pi - \theta_{i-1} < \angle p_{i-1} p_i p_{i+1}$, $P'$ lies inside $P$ at $p_i$. 
    %$\Delta_{i-1}^t$ and $\Delta_i^t$ are: (a) disjoint; (b) overlapping. 
    }
    \label{fig:101}
\end{figure}

Thus we must have the situation shown in Figure~\ref{fig:101}(b), so 
$\Delta_{i-1}^t$ and $\Delta_i^t$ are disjoint 
and $P^t$ remains convex at $p_i^t$.
}
\end{proof}

With these claims in hand, we can complete the proof of the theorem.  
Divide the circular sequence $\Delta_1^t, \ldots, \Delta_n^t$ into maximal subsequences all of the same type (all 01 or all 10).  If $D_{i,j} = \Delta_i^t, \ldots, \Delta_j^t$ is such a maximal subsequence then by Claims~\ref{claim:0101} and~\ref{claim:1010}
no two triangles of $D_{i,j}$ intersect, and all the triangles of $D_{i,j}$ live in the region $R_{i,j}$ bounded by $p_i^t, \ldots, p_{j+1}^t$ and two bounding rays---\changed{$r_i^0$ and $r_j^1$} in the case of 01 triangles, as shown in Figure~\ref{fig:convex-morph-config-2}.
Between one sequence $D_{i,j}$ and the next, $D_{j+1,l}$, Claim~\ref{claim:101} implies that the regions $R_{i,j}$ and $R_{j+1,l}$ are disjoint. 
%\attention{Claim 9 is about triangles.  Argue that regions are disjoint.}
\end{proof}

\begin{figure}[htb]
    \centering
    \includegraphics[width=.7\columnwidth]{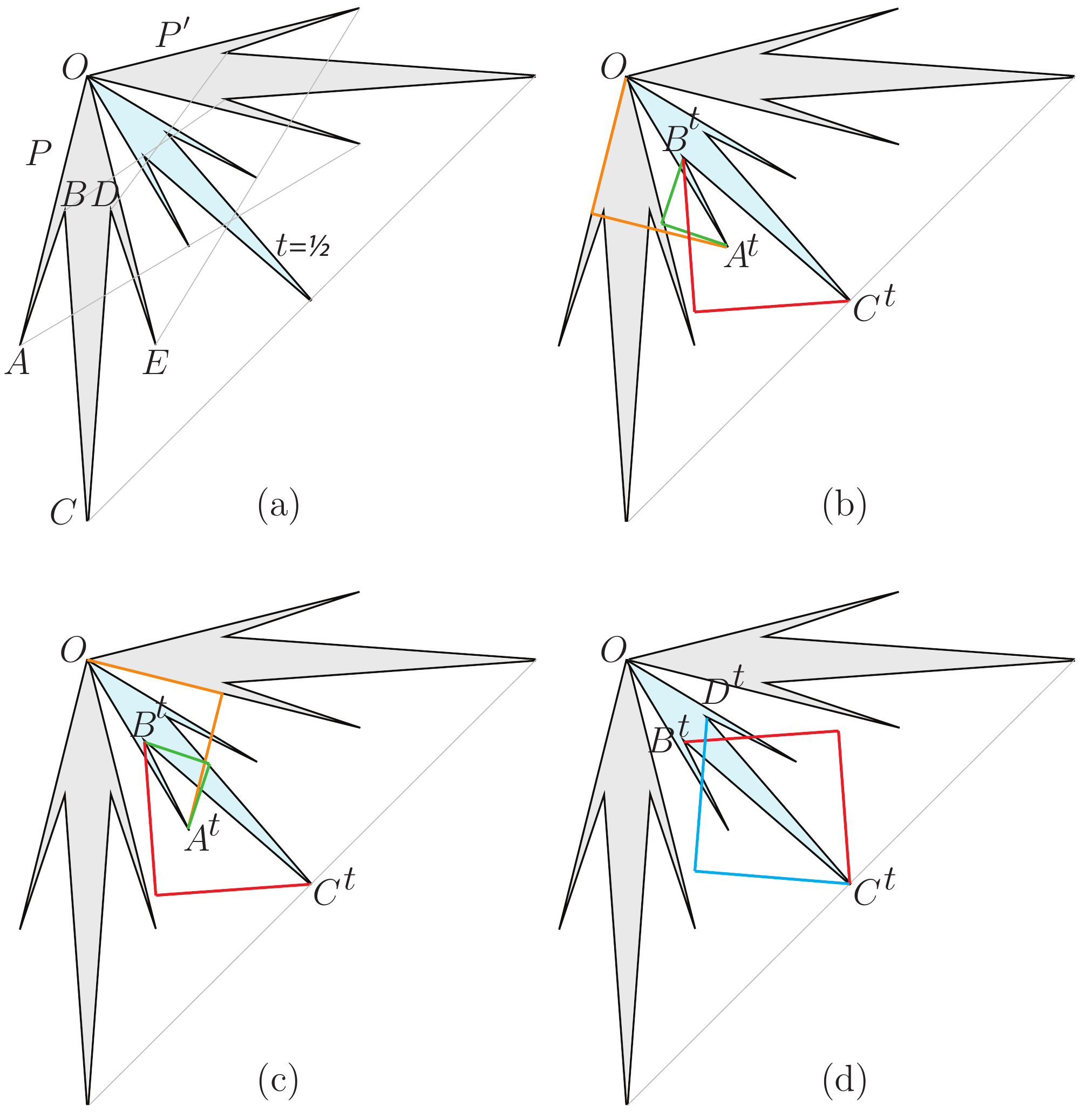}
    \caption{A star-shaped polygon $P$ where the linear morph to its $90^\circ$ rotation $P'$ preserves planarity, but there is no banded surface: (a) $P$ and $P'$ and the intermediate position of the linear morph at $t = \frac{1}{2}$  (shaded blue); (b) If we choose 
    \changed{the chord for edge $BC$ that produces the ``outward'' triangle (shown  in red)}
    %the ``outward'' chord for edge $BC$ 
    then at $t = \frac{1}{2}$
    it intersects one choice for $OA$ and one choice for $BA$; (c) The other choices for $OA$ and $BA$ intersect each other; (d) Thus we are forced to
    \changed{choose the chord for edge $BC$ that produces the ``inward'' triangle (shown in red)},
    %use the ``inward'' chord for edge $BC$, 
    and, by symmetry, the ``inward'' \changed{triangle} %chord
    for $DC$ \changed{(shown in cyan)}---but these intersect.
    %Note from Anna: it's the identical polygon that Therese drew, just expanded into several frames, and coloured.
    }
    \label{fig:star}
\end{figure}

%%%%%%%%%%%%%%%%%%%%%%%%%%%%%%%%%%%%%
% Conclusions
%%%%%%%%%%%%%%%%%%%%%%%%%%%%%%%%%%%%%
\section{Conclusions}

We have introduced the idea of a banded surface to construct a polyhedron between two polygons in parallel slices and have explored some connections between linear morphs and banded surfaces without Steiner points.
Many  questions remain,  the  two main  ones being:

\begin{enumerate}     
    \item Is there a bound  better than $O(n^2)$ on the number of Steiner  points needed to construct a banded surface between  two $n$-vertex polygons?
    %What is a bound on the number of Steiner points needed to construct a banded surface between two $n$-vertex polygons?
    What if the polygons are convex?
    
    %Can Steiner points be limited to the polygon boundaries? or to  the plane $z=\frac{1}{2}$?   
%    A starting point would be to 
%    characterize when there is a banded surface without Steiner points.
%    (We have a poly time algorithm to test this, but it might be more useful to characterize them some other way.)  % What about for triangles?  

%    \item When does it suffice to put Steiner points at $z=0$ and $z=1$?  Maybe if there is a planarity-preserving linear morph from $P$ to $P'$ (where the polygons are not necessarily convex) then adding Steiner points along the edges of $P$ and $P'$ makes the banded surface ``hug'' the linear morph surface.
    
    \item Is it NP-hard to minimize the number of Steiner points needed to construct a banded surface? 
    
%    \item Does Theorem~\ref{theorem:convex-morph-revised} extend to the case where $P$ and $P'$ are convex and the linear morph between them preserves planarity but not convexity?  We conjecture that it does.

\end{enumerate}

%---------------------------- Bibliography -------------------------------

% Please add the contents of the .bbl file that you generate,  or add bibitem entries manually if you like.
% The entries should be in alphabetical order
\small
\bibliographystyle{abbrv}

%\begin{thebibliography}{99}
\bibliography{reconstruct}
%\end{thebibliography}

\end{document}